\documentclass[11pt,a4paper]{article}
\usepackage{amsmath,amssymb,amsthm,geometry}
\usepackage{mathrsfs}
\newcommand{\HH}{\mathbb{H}}

\newcommand{\RR}{\mathbb{R}}

\newcommand{\dLam}{d\Lambda_{\HH}}

\newcommand{\pB}{{\partial \mathbb{B}_{\mathbb{H}}^n}}
\newcommand{\cB}{{\overline{\mathbb{B}_{\mathbb{H}}^n}}}

\DeclareMathOperator{\arctanh}{arctanh}

\DeclareMathOperator{\Vol}{Vol}

\newcommand{\BH}{\mathbb{B}_{\mathbb{H}}^n}

\newcommand{\br}{\mathbb{R}}
\newcommand{\bh}{\mathbb{H}}

\newtheorem{theorem}{Theorem}[section]
\newtheorem{definition}{Definition}[section]
\newtheorem{lemma}{Lemma}[section]
\newtheorem{proposition}{Proposition}[section]
\newtheorem{corollary}{Corollary}[section]
\newtheorem{remark}{Remark}[section]
\newtheorem{example}{Example}[section]

\theoremstyle{definition}

\title{Conformal Barycenters in Quaternionic Hyperbolic Balls}
\author{ Wensheng Cao  and   Zhijian Ge}
\date{}

\begin{document}
	\maketitle
	
	\begin{abstract}
		We extend the notion of conformal barycenter, recently introduced by
		Ja\v{c}imovi\'{c} and Kalaj for the complex hyperbolic ball, to the
		quaternionic unit ball $\BH$.  The quaternionic
		conformal barycenter of a measurable set $D$ with finite hyperbolic
		measure and finite first moment is defined as the unique point $c$ such that
		$\int_D \Phi_c(q)\, \dLam(q) = \mathbf{0}$, where
		$\Phi_c$ is the quaternionic Hua involution exchanging $0$ and $c$.
		Equivalently, it is the unique minimum of the energy functional
		$G(x) = \int_D \log\cosh^2\!\big(\frac12 d_H(x,y)\big)\, \dLam(y)$.
		We prove existence and uniqueness using the strict geodesic convexity
		of $G$, which is established by a direct computation along geodesics.
		The barycenter is invariant under
		the full isometry group $\mathrm{Sp}(n,1)$.  We also treat finite
		point sets and provide explicit examples.
	\end{abstract}
	
	\medskip
	\noindent\textbf{2020 Mathematics Subject Classification:} 51M10, 53C35, 15B33.
	
	\medskip
	\noindent\textbf{Keywords:} quaternionic hyperbolic space, conformal barycenter,
	Hua transformation, Bergman metric, geodesic convexity.
	
	\section{Introduction}
	
	The barycenter, or center of mass, is a fundamental notion in
	Euclidean geometry.  For a measurable set $D \subset \mathbb{R}^m$
	with finite Lebesgue measure, the Euclidean barycenter is simply
	the average $\frac{1}{|D|} \int_D x \, dx$.  This definition relies
	crucially on the linear structure of Euclidean space and is not
	invariant under the conformal transformations of the sphere or
	the hyperbolic ball.
	
	In a seminal 1986 paper, Douady and Earle \cite{DE86} introduced
	the notion of a conformal barycenter for probability measures on
	the circle $S^1$.  With each measure they associated a vector field
	on the unit disk, and the unique zero of this field was called the
	conformal barycenter.  This construction was conformally invariant
	and led to a canonically defined extension of quasi-symmetric
	homeomorphisms of the circle.
	
	Recently, Ja\v{c}imovi\'{c} and Kalaj \cite{JaKal25} introduced
	the notions of conformal and holomorphic barycenters for measurable
	sets inside the hyperbolic ball.  For the Poincar\'{e} ball
	$\mathbb{B}^n \subset \mathbb{R}^n$, the conformal barycenter is the
	unique point $c$ such that $\int_D h_c(x)\, d\Lambda(x) = 0$, where
	$h_c$ is the M\"obius involution exchanging $0$ and $c$, and
	$d\Lambda$ is the hyperbolic volume measure.  For the Bergman ball
	$\mathbb{B}_m \subset \mathbb{C}^m$, the analogous holomorphic
	barycenter uses the involutive automorphism $p_c$ and the
	Bergman metric.  These two notions coincide in the complex disk
	($m=1$, real dimension $2$) but differ in higher dimensions.
	
	In the present paper, we extend this construction to the quaternionic
	unit ball $\BH$.  Quaternionic hyperbolic
	geometry has been intensively studied; see Chen--Greenberg \cite{chen},
	Kim--Parker \cite{kimp}, and the authors' works on volumes and
	discrete groups \cite{CG2011,cao-congruence,cao16,caofu}.  The symmetry group is the
	non-compact symplectic group $\mathrm{Sp}(n,1)$, which acts by
	isometries of the Bergman metric.  Due to the non-commutativity of
	the quaternions, some formulas require more careful treatment than
	in the complex case, but the essential geometric structure, namely
	negative sectional curvature, the existence of involutive
	isometries $\Phi_u$, and the strict convexity of the distance
	function, carries over perfectly.
	
	Our main result is:
	
	\begin{theorem}[Main Theorem]\label{thm:main}
		Let $D \subseteq \BH$ be measurable with
		$0 < \Lambda_{\HH}(D) < \infty$ and $\int_D d_H(0,y)\,\dLam(y) < \infty$.
		\begin{enumerate}
			\item[(i)] There exists a unique point $c = c(D) \in \BH$, called the
			\emph{quaternionic conformal barycenter} of $D$, such that
			\[
			\int_D \Phi_c(q) \; \dLam(q) = \mathbf{0}.
			\]
			\item[(ii)] The point $c$ is the unique global minimum of the
			energy functional
			\[
			G(x) = \int_D \log\cosh^2\!\Bigl(\frac12 d_H(x,y)\Bigr)
			\; \dLam(y).
			\]
			\item[(iii)] For any $g \in \mathrm{Sp}(n,1)$,
			$c(g(D)) = g(c(D))$.
		\end{enumerate}
	\end{theorem}
	
	The proof of parts~(i) and~(ii) is carried out in
	Theorem~\ref{thm:existence} (Section~\ref{sec:proof-main}), and
	part~(iii) is proved in Theorem~\ref{thm:invariance} within the
	same subsection.
	
	The paper is organized as follows.  In Section~\ref{sec:qhs} we recall
	the necessary background on quaternionic hyperbolic geometry.
	Section~\ref{sec:metric-measure} discusses the invariant
	metric and measure on the ball model, including an explicit
	computation of the volume form and a parametrisation of geodesics.
	Section~\ref{sec:Phi} constructs the
	quaternionic Hua involution $\Phi_u$ and proves the intertwining
	relation and the Jacobian formula.  Section~\ref{sec:barycenter}
	contains the definition of the quaternionic barycenter and the proof
	of the main theorem.  Section~\ref{sec:finite} treats finite point
	sets, and Section~\ref{sec:examples} provides examples illustrating
	the theorem.
	
	\section{Quaternionic hyperbolic space}\label{sec:qhs}
	
	In this section, we give necessary background materials on
	quaternionic hyperbolic geometry.  More details can be found in
	\cite{chen,gold,kimp,most}.
	
	We recall that a real quaternion is of the form
	$q = q_0 + q_1\mathbf{i} + q_2\mathbf{j} + q_3\mathbf{k} \in \bh$
	where $q_i \in \br$ and
	$\mathbf{i}^2 = \mathbf{j}^2 = \mathbf{k}^2 = \mathbf{i}\mathbf{j}\mathbf{k} = -1$.
	Let $\overline{q} = q_0 - q_1\mathbf{i} - q_2\mathbf{j} - q_3\mathbf{k}$
	and $|q| = \sqrt{\overline{q}q} = \sqrt{q_0^2 + q_1^2 + q_2^2 + q_3^2}$
	be the conjugate and modulus of $q$, respectively.
	
	Throughout this paper, we treat $\bh^{n+1}$ as a {\it right vector space}
	over $\bh$.  Thus scalar multiplication is written as $v q$ for
	$v \in \bh^n$, $q \in \bh$, and linear transformations act on the left.

	Let $\bh^{n,1}$ be the vector space of dimension $n+1$ over $\bh$
	with the unitary structure defined by the Hermitian form
	\begin{equation}\label{eq:herm-form}
		\langle{\bf z},\,{\bf w}\rangle = {\bf w}^* J {\bf z} =
		\overline{w_1}z_1 + \cdots + \overline{w_n}z_n - \overline{w_{n+1}}z_{n+1},
	\end{equation}
	where ${\bf z}$ and ${\bf w}$ are column vectors in $\bh^{n,1}$ with
	entries $(z_1,\dots,z_{n+1})$ and $(w_1,\dots,w_{n+1})$ respectively,
	$\cdot^*$ denotes the conjugate transpose, and $J$ is the Hermitian
	matrix
	\[
	J = \begin{pmatrix}
		I_n & 0 \\
		0 & -1
	\end{pmatrix}.
	\]
	
	For vectors $z,w \in \bh^n$, we use the same notation $\langle \cdot,\cdot \rangle$
	for the standard quaternionic Hermitian inner product
	$\langle z, w\rangle = w^* z$; this is a slight abuse of notation, but it
	can be regarded as the restriction of the Hermitian form \eqref{eq:herm-form}
	to vectors of the form $(z,0)$, i.e.\ $\langle z,w\rangle = \langle (z,0), (w,0)\rangle$.
	In coordinates,
	$\langle z,w\rangle = \overline{w_1}z_1+\cdots+\overline{w_n}z_n$.
	It satisfies $\overline{\langle z,w\rangle} = \langle w,z\rangle$ and is
	\emph{linear} in the first argument and conjugate linear in the second,
	because for any $\lambda,\mu\in\mathbb{H}$,
	$\langle z \lambda, w \mu\rangle = \overline{\mu}\,\langle z,w\rangle\,\lambda$.
	
	We define a \emph{unitary transformation} $g$ to be an automorphism
	of $\bh^{n,1}$, that is, a linear bijection such that
	$\langle g({\bf z}),\, g({\bf w})\rangle = \langle{\bf z},\,{\bf w}\rangle$
	for all ${\bf z}$ and ${\bf w}$ in $\bh^{n,1}$.  The group of all
	unitary transformations is denoted by $\mathrm{Sp}(n,1)$:
	\begin{equation}\label{eq:Sp-def}
		\mathrm{Sp}(n,1) = \{ A \in \mathrm{GL}(n+1,\bh) : A^* J A = J \}.
	\end{equation}
	
	Following \cite{chen}, set
	\begin{align*}
		V_0 &= \{ {\bf z} \in \bh^{n,1} \setminus \{0\} :
		\langle{\bf z},\,{\bf z}\rangle = 0 \}, \\
		V_- &= \{ {\bf z} \in \bh^{n,1} :
		\langle{\bf z},\,{\bf z}\rangle < 0 \}.
	\end{align*}
	Let $V^s = V_- \cup V_0$ and let
	$P : V^s \to P(V^s) \subset \bh^n$ be the projection map
	\[
	P(z_1,\dots,z_n, z_{n+1})^{\mathrm t}
	= (z_1 z_{n+1}^{-1}, \dots, z_n z_{n+1}^{-1})^{\mathrm t},
	\]
	where $\cdot^{\mathrm t}$ denotes the transpose.  Note that we use right
	division by $z_{n+1}$, which is consistent with the right vector space
	structure.
	
	Quaternionic hyperbolic $n$-space is
	$\mathbf{H}_{\bh}^n = P(V_-)$,
	and its boundary is $\partial \mathbf{H}_{\bh}^n = P(V_0)$.
	In non-homogeneous coordinates, we identify $\mathbf{H}_{\bh}^n$ with
	the unit ball
	\[
	\BH= \bigl\{z\in\mathbb{H}^{n}:|z_{1}|^{2}+\cdots+|z_{n}|^{2}<1\bigr\}.
	\]

	The Bergman metric on $\mathbf{H}_{\bh}^n$ is given by the distance
	formula
	\begin{equation}\label{eq:dist-proj}
		\cosh^2 \frac{\rho(z,w)}{2}
		= \frac{\langle{\bf z},\,{\bf w}\rangle \langle{\bf w},\,{\bf z}\rangle}
		{\langle{\bf z},\,{\bf z}\rangle \langle{\bf w},\,{\bf w}\rangle},
		\qquad {\bf z} \in P^{-1}(z),\; {\bf w} \in P^{-1}(w).
	\end{equation}
	Alternatively, the Bergman metric is given by the line element
	\begin{equation}\label{eq:metric-line}
		ds^2 = \frac{-4}{\langle{\bf z},\,{\bf z}\rangle} \det
		\begin{pmatrix}
			\langle{\bf z},\,{\bf z}\rangle & \langle d{\bf z},\,{\bf z}\rangle \\
			\langle{\bf z},\,d{\bf z}\rangle & \langle d{\bf z},\,d{\bf z}\rangle
		\end{pmatrix}.
	\end{equation}
	(The determinant is understood via the real representation of quaternions.)
	
	The holomorphic isometry group of $\mathbf{H}_{\bh}^n$ with respect
	to the Bergman metric is the projective unitary group
	$\mathrm{PSp}(n,1) = \mathrm{Sp}(n,1)/\{\pm I_{n+1}\}$,
	which acts on $P(\bh^{n,1})$ by matrix multiplication.
	
	Let  $g = \begin{pmatrix} A & \alpha \\ \beta & a \end{pmatrix} \in \mathrm{Sp}(n,1)$, where $A \in \mathbb{H}^{n \times n},\,\alpha \in \mathbb{H}^{n},\, \beta \in \mathbb{H}^{1 \times n},\,a \in \mathbb{H}$. The preservation of the Hermitian form implies $J=g^{*}Jg$; consequently
	\begin{equation}\label{ginv}
		g^{-1}=Jg^{*}J=\begin{pmatrix}A^{*}&-\beta^{*}\\-\alpha^{*}&\bar a\end{pmatrix}.
	\end{equation}
	The identities $gg^{-1}=g^{-1}g=I$ yield the standard relations
	for elements of $\mathrm{Sp}(n,1)$; in particular we have
	$|a|^{2}=|\alpha|^{2}+1$ and $|a|^{2}=|\beta|^{2}+1$, which
	imply $|a|\geq 1$.
	
	The action of  $g$  on the ball model is 
	\[
	g(z) = (Az + \alpha)(\beta z + a)^{-1},\quad z\in \cB.
	\]

	Following Chen and Greenberg \cite{chen}, a non-trivial element
	$g \in \mathrm{Sp}(n,1)$ is called:
	\begin{itemize}
		\item[(i)] \emph{elliptic} if it has a fixed point in $\mathbf{H}_{\bh}^n$;
		\item[(ii)] \emph{parabolic} if it has exactly one fixed point, which
		lies in $\partial \mathbf{H}_{\bh}^n$;
		\item[(iii)] \emph{loxodromic} if it has exactly two fixed points,
		both lying in $\partial \mathbf{H}_{\bh}^n$.
	\end{itemize}
	
	Let
	\[
	\mathrm{Sp}(n) = \{ A \in \mathrm{GL}(n,\bh) : A A^* = I_n \}, \qquad
	\mathrm{Sp}(1) = \{ q \in \bh : |q| = 1 \}.
	\]
	The stabilizer of the point with homogeneous coordinates
	$(0,\dots,0,1)$ is
	\[
	K = \mathrm{Sp}(n) \times \mathrm{Sp}(1)
	= \left\{ \begin{pmatrix}
		A & 0 \\
		0 & q
	\end{pmatrix} : A \in \mathrm{Sp}(n),\; q \in \mathrm{Sp}(1) \right\}.
	\]
	We identify the quaternionic hyperbolic space with the homogeneous space
	\begin{equation}\label{eq:homogeneous}
		\mathbf{H}_{\bh}^n = \mathrm{Sp}(n,1) / (\mathrm{Sp}(n) \times \mathrm{Sp}(1)).
	\end{equation}
	
	\section{Invariant metric and measure on the ball model}\label{sec:metric-measure}
	
	For computations, it is convenient to work directly with the unit
	ball model $\BH$.  The hyperbolic metric (see \cite[Proposition 4.1]{cao16})   given by (\ref{eq:metric-line})   can be expressed as
	\begin{equation}\label{eq:ball-metric}
		ds^2 = \frac{4|dq|^2}{1-|q|^2}
		+ \frac{4\bigl|\sum_{i=1}^n \overline{q_i}\,dq_i\bigr|^2}
		{(1-|q|^2)^2},
	\end{equation}
	where the sum is the usual quaternionic Hermitian product on $\bh^n$.

	The distance function $d_H$ associated with~\eqref{eq:ball-metric} can be
	expressed through the involutive automorphism $\Phi_q$ constructed in
	Section~\ref{sec:Phi}.  Since $\Phi_q(q)=0$ and $\Phi_q$ is an isometry,
	\[
	d_H(p,q)=d_H(\Phi_q(p),\Phi_q(q))=d_H(\Phi_q(p),0)
	=\log\frac{1+|\Phi_q(p)|}{1-|\Phi_q(p)|}.
	\]
	Thus
	\begin{equation}\label{eq:dist-Phi}
		d_H(p,q) = \log \frac{1 + |\Phi_q(p)|}{1 - |\Phi_q(p)|}.
	\end{equation}
	The isometry group $\mathrm{Sp}(n,1)$ acts transitively on $\BH$ and
	preserves $d_H$.
	
	The Riemannian volume form of the metric \eqref{eq:ball-metric} can
	be obtained by a direct computation.
	Write $z=(z_1,\dots,z_n)\in\BH$, $r=|z|$, and
	$z_i=x_i+y_i\mathbf{i}+v_i\mathbf{j}+w_i\mathbf{k}$.
	At a point $P=(0,\dots,0,r)$ the real $4n\times4n$ metric tensor
	takes the diagonal form
	\[
	g_{\mathbb{R}}|_P = \frac{4}{(1-r^2)^2}\,
	\operatorname{diag}(1-r^2,\dots,1-r^2,\,1,1,1,1),
	\]
	whose determinant is $\det g_{\mathbb{R}}=4^{4n}/(1-r^2)^{4(n+1)}$.
	By homogeneity this formula holds everywhere.
	Hence the Riemannian volume element is
	\begin{equation}\label{eq:vol-form-unnormalized}
		d\Vol = \frac{4^{2n}}{(1-r^2)^{2(n+1)}}\, r^{4n-1}\,dr\,d\sigma_{4n-1},
	\end{equation}
	where $d\sigma_{4n-1}$ is the Euclidean volume element of the
	unit sphere $S^{4n-1}$.
	
	The metric~\eqref{eq:ball-metric} already has holomorphic sectional
	curvature $-1$, so its Riemannian
	volume form is the hyperbolic measure $\dLam$ on $\BH$.
	From~\eqref{eq:vol-form-unnormalized} we obtain, in polar coordinates,
	\[
	\dLam(q) = \frac{4^{2n}}{(1-r^2)^{2(n+1)}}\, r^{4n-1}\,dr\,d\sigma_{4n-1},
	\]
	and, in Cartesian coordinates,
	\begin{equation}\label{eq:hyp-measure}
		\dLam(q) = \frac{4^{2n}\,d\lambda(q)}{(1-|q|^2)^{2n+2}},
	\end{equation}
	where $d\lambda(q)$ is the Lebesgue measure on $\HH^n\cong\RR^{4n}$.
	
	Using $r=\tanh(\rho/2)$, the hyperbolic volume of a ball
	$B(\rho)=\{q\in\BH:d_H(0,q)<\rho\}$ is
	\begin{align}
		\Vol(B(\rho))
		&= \int_0^{\tanh(\rho/2)}\! \frac{4^{2n}\,\sigma_{4n-1}\,
			r^{4n-1}}{(1-r^2)^{2(n+1)}}\,dr \nonumber\\
		&= \frac{(4\pi)^{2n}}{(2n+1)!}
		\sinh^{4n}\!\Bigl(\frac{\rho}{2}\Bigr)
		\Bigl[1 + 2n\cosh^2\!\Bigl(\frac{\rho}{2}\Bigr)\Bigr].
		\label{eq:ball-volume}
	\end{align}
	The detailed computation can be found in~\cite{cao16}.
	From~\eqref{eq:ball-volume} we obtain
	$\operatorname{Vol}(B(\rho))\asymp e^{(2n+1)\rho}$ as $\rho\to\infty$.
	The exponent $2n+1$ is typical of rank-one symmetric spaces of non-compact
	type.  This exponential volume growth, together with the finite
	first-moment assumption, guarantees that the energy functional $G(x)$ grows
	fast enough as $|x|\to1$ to ensure a global minimum (see
	Section~\ref{sec:barycenter}).

	For the convexity argument we need a convenient parametrisation of
	geodesics in the ball model.  The following description follows
	directly from the homogeneous model recalled in Section~2.
	
	\begin{proposition}[Geodesics through the origin]\label{prop:geod-origin}
		Let $v\in\pB$ be a unit vector (i.e.\ $|v|=1$).
		Then the curve
		\[
		\gamma(t)=\tanh\!\Bigl(\frac{t}{2}\Bigr)v,\qquad t\in\RR,
		\]
		is a geodesic in $\BH$ parametrised by arc length $t$, with $\gamma(0)=0$.
		Conversely, every geodesic through the origin can be written in this form,
		up to composition with a linear isometry $U\in\mathrm{Sp}(n)\times\mathrm{Sp}(1)$.
	\end{proposition}
	
	\begin{proof}
		Lift $v$ to the null vectors
		\[
		\tilde{\mathbf u}=\frac{1}{\sqrt{2}}(v,1),\qquad
		\tilde{\mathbf v}=\frac{1}{\sqrt{2}}(-v,1)\in V_0,
		\]
		which satisfy $\langle\tilde{\mathbf u},\tilde{\mathbf v}\rangle=-1$.
		The curve
		$\Gamma(t)=e^{t/2}\tilde{\mathbf u}+e^{-t/2}\tilde{\mathbf v}$
		projects to
		\[
		P(\Gamma(t))=\frac{(e^{t/2}-e^{-t/2})v}{e^{t/2}+e^{-t/2}}
		=\tanh\!\Bigl(\frac t2\Bigr)v,
		\]
		and is a geodesic parametrised by arc length because
		$\langle\Gamma(t),\Gamma(t)\rangle=-2$ and
		$d_H(0,\gamma(t))=\log\frac{1+|\gamma(t)|}{1-|\gamma(t)|}=|t|$.
		Conversely, a geodesic through $0$ has antipodal boundary points $u,-u$;
		a linear isometry sending $u$ to $v$ maps it to the curve above.
	\end{proof}
	
	\section{The involutive automorphism $\Phi_u$}\label{sec:Phi}
	
	We construct the involutive automorphism $\Phi_u$ that exchanges
	$0$ and $u$, the quaternionic analogue of the Hua transformation
	in the complex case \cite{Hua,Rudin,zhu}.  Due to the non-commutativity of the
	quaternions, careful handling of the right module structure is
	required.
	
	\subsection{Quaternionic Hua transformation}
	
	For $u \in \BH$ with $|u| < 1$, set $s = \sqrt{1-|u|^2}$.
	Define the matrix
	\begin{equation}\label{eq:Au}
		A_u = \frac{1}{1+s}\,uu^{*} + sI_n .
	\end{equation}
	Here $uu^{*}$ denotes the quaternionic outer product, acting on a vector $x$ as $u\langle x,u\rangle = u(u^{*}x)$.
	When $u = 0$ we have $s = 1$, and the formula gives $A_0 = I_n$.
	
	\begin{lemma}\label{lemmaAu}
		The matrix $A_u$ satisfies:
		\begin{enumerate}
			\item[(i)] $A_u u = u$, and $A_u v = s\,v$ for all $v \perp u$.
			\item[(ii)] $A_u^{-1} = -\dfrac{1}{(1+s)s}\,uu^{*} + \dfrac1s I_n$.
			\item[(iii)] $A_u$ is Hermitian, positive definite, and
			$A_u^2 = s^2 I_n + uu^{*}$.
		\end{enumerate}
	\end{lemma}
	
	\begin{proof}
		(i) Because $u^*u = |u|^2 = 1-s^2$, we have
		$A_u u = \frac{1}{1+s}u(u^*u) + s u
		= \frac{1-s^2}{1+s}u + s u = (1-s)u + s u = u$.
		If $v \perp u$, i.e.\ $u^* v = 0$, then $A_u v = s v$.
		(ii) A direct verification shows that
		$\bigl(-\frac{1}{(1+s)s}uu^* + \frac1s I_n\bigr) A_u = I_n$.
		(iii) $A_u^* = \frac{1}{1+s} uu^* + s I_n = A_u$, so $A_u$ is Hermitian. Since \(x^{*}(uu^{*})x = |u^{*}x|^{2}\ge 0\) for every \(x\in\mathbb{H}^{n}\), the matrix \(uu^{*}\) is positive semidefinite, and thus \(A_{u}\) is positive definite.
		Finally,
		\begin{align*}
			A_u^2 &= \bigl( s I_n + c\, uu^* \bigr)^2,\quad c=\frac{1}{1+s}\\
			&= s^2 I_n + \bigl(2sc + c^2|u|^2\bigr) uu^* \\
			&= s^2 I_n + \bigl(2s\frac{1}{1+s} + \frac{1}{(1+s)^2}(1-s^2)\bigr) uu^* \\
			&= s^2 I_n + uu^*.
		\end{align*}
	\end{proof}
	
	The quaternionic Hua involution is the map
	$\Phi_u : \overline{\BH} \to \overline{\BH}$ defined by the
	projective action of the matrix
	\begin{equation}\label{defphi}
		\Phi_u = \begin{pmatrix}
			-\dfrac{A_u}{s} & \dfrac{u}{s} \\[10pt]
			-\dfrac{u^{*}}{s} & \dfrac{1}{s}
		\end{pmatrix} \in \mathrm{Sp}(n,1).
	\end{equation}
	Explicitly, for $z \in \overline{\BH}$,
	\begin{equation}\label{eq:Phi-def}
		\Phi_u(z) = \bigl( u - A_u z \bigr)
		\bigl(1 - \langle z, u \rangle \bigr)^{-1}.
	\end{equation}
	(Recall that $\langle z,u\rangle = u^*z$ with the inner product defined
	in Section~2.)
	For $u = 0$ we have $A_0 = I_n$, $s = 1$, and $\Phi_0(z) = -z$.
	
	\begin{proposition}[Properties of $\Phi_u$]\label{prop:Phi-prop}
		\begin{itemize}
			\item[(i)] $\Phi_u(0) = u$ and $\Phi_u(u) = 0$.
			\item[(ii)] $\Phi_u \circ \Phi_u = \operatorname{id}$.
			\item[(iii)] $\displaystyle
			|\Phi_u(z)|^2 = 1 - \frac{(1-|u|^2)(1-|z|^2)}
			{|1 - \langle z, u \rangle|^2}$.
			\item[(iv)] $\Phi_u$ maps $\BH$ onto itself and extends
			continuously to the boundary.
			\item[(v)] $\Phi_u^2 = I_{n+1}$ and $\Phi_u = \Phi_u^{-1}$.
			\item[(vi)] The matrix $\Phi_u$ has the following eigenvectors:
			\[
			\Phi_u\begin{pmatrix}u\\1+s\end{pmatrix}
			= \begin{pmatrix}u\\1+s\end{pmatrix},\qquad
			\begin{pmatrix}u\\1+s\end{pmatrix}\in V_-,
			\]
			\[
			\Phi_u\begin{pmatrix}u\\1-s\end{pmatrix}
			= -\begin{pmatrix}u\\1-s\end{pmatrix},\qquad
			\begin{pmatrix}u\\1-s\end{pmatrix}\in V_+,
			\]
			and for every $v\perp u$,
			\[
			\Phi_u\begin{pmatrix}v\\0\end{pmatrix}
			= -\begin{pmatrix}v\\0\end{pmatrix},\qquad
			\begin{pmatrix}v\\0\end{pmatrix}\in V_+.
			\]
			Thus $\Phi_u$ is an involutive elliptic isometry whose unique
			fixed point in $\BH$ is $u/(1+s)$.  Geometrically, $\Phi_u$
			is the geodesic symmetry at the point $u/(1+s)$.
		\end{itemize}
	\end{proposition}
	
	\begin{proof}
		(i) Since $A_u 0 = 0$, we have $\Phi_u(0) = u \cdot 1^{-1} = u$.
		For $\Phi_u(u)$, note that $A_u u = u$ by Lemma~\ref{lemmaAu}(i),
		and $\langle u,u\rangle = |u|^2$, so
		$\Phi_u(u) = (u - u)(1 - |u|^2)^{-1} = 0$.
		
		(ii) We verify that $\Phi_u$ is an involution by a direct matrix computation.
		Using the block form
		\[
		\Phi_u = \frac1s \begin{pmatrix} -A_u & u \\ -u^* & 1 \end{pmatrix},
		\]
		we obtain
		\[
		\Phi_u^2 = \frac1{s^2}
		\begin{pmatrix}
			A_u^2 - u u^* & -A_u u + u \\
			u^* A_u - u^* & -u^* u + 1
		\end{pmatrix}.
		\]
		Now $A_u u = u$, $u^* A_u = (A_u u)^* = u^*$, and
		$u^* u = |u|^2 = 1-s^2$.  Moreover, by Lemma~\ref{lemmaAu}(iii),
		$A_u^2 - u u^* = s^2 I_n$.  Hence
		\[
		\Phi_u^2 = \frac1{s^2}
		\begin{pmatrix}
			s^2 I_n & 0 \\
			0 & 1 - (1-s^2)
		\end{pmatrix}
		= I_{n+1}.
		\]
		Thus $\Phi_u \circ \Phi_u = \operatorname{id}$.
		
		(iii) Using $A_u z = \frac{1}{1+s}u\langle u,z\rangle + s z$ and
		decomposing $z$ into components parallel and perpendicular to $u$,
		one obtains after simplification
		\[
		|u - A_u z|^2 = |1 - \langle z,u\rangle|^2 - (1-|u|^2)(1-|z|^2).
		\]
		Dividing by $|1 - \langle z,u\rangle|^2$ gives the formula.
		
		(iv) From (iii), $|\Phi_u(z)| < 1 \iff |z| < 1$, so $\Phi_u$
		maps $\BH$ bijectively onto itself and extends continuously.
		
		(v) This is exactly the matrix identity $\Phi_u^2 = I_{n+1}$ proved in (ii).
		
		(vi) The eigenvector relations follow by direct computation using
		$A_u u = u$ and $A_u v = s v$ for $v \perp u$.
	\end{proof}
	
	\begin{remark}
		With respect to the Bergman metric,
		$\Phi_u$ acts as an isometry of $\BH$, since it belongs to
		$\mathrm{Sp}(n,1)$.
	\end{remark}

	\begin{theorem}[Intertwining relation]\label{thm:intertwining}
		For every $g \in \mathrm{Sp}(n,1)$ and every $c \in \BH$,
		\[
		\Phi_{g(c)} \circ g = U \circ \Phi_c,
		\]
		where $U \in \mathrm{Sp}(n) \times \mathrm{Sp}(1)$ is a linear isometry
		of $\BH$ fixing the origin.
	\end{theorem}
	
	\begin{proof}
		Consider the map $U = \Phi_{g(c)} \circ g \circ \Phi_c$.
		Using Proposition~\ref{prop:Phi-prop}(i) and the fact that
		$\Phi_c = \Phi_c^{-1}$,
		\[
		U(0) = \Phi_{g(c)}\bigl(g(\Phi_c(0))\bigr)
		= \Phi_{g(c)}(g(c)) = 0.
		\]
		Thus $U$ fixes the origin.  Since each factor is an isometry,
		$U$ is an isometry fixing the origin.  By the homogeneous space
		identification \eqref{eq:homogeneous}, the isotropy group of the
		origin is $\mathrm{Sp}(n) \times \mathrm{Sp}(1)$. Therefore $U \in \mathrm{Sp}(n) \times \mathrm{Sp}(1)$ is a linear
		isometry, i.e.\ it acts by a linear transformation on the ball model.
		Composing with $\Phi_c$ on the right yields the stated relation.
	\end{proof}
	
	\subsection{Jacobian and invariance of the measure}
	
	\begin{proposition}[Jacobian of $\Phi_u$]\label{prop:Jacobian}
		The real Jacobian determinant of $\Phi_u : \BH \to \BH$
		satisfies
		\[
		|\det D_{\mathbb{R}}\Phi_u(z)|
		= \frac{(1-|u|^2)^{2n+2}}{|1 - \langle z, u \rangle|^{4n+4}}.
		\]
	\end{proposition}
	\begin{proof}
		Since $\Phi_u \in \mathrm{Sp}(n,1)$ is an isometry of the Bergman metric,
		it preserves the Riemannian volume form.  In particular, the hyperbolic
		measure $\dLam$ defined in~\eqref{eq:hyp-measure} is invariant under
		$\Phi_u$, i.e.\ $\Phi_u^*(\dLam) = \dLam$.  Writing this in Cartesian
		coordinates and cancelling the common factor $4^{2n}d\lambda(z)$ gives
		\[
		\frac{|\det D_{\mathbb{R}}\Phi_u(z)|}{(1 - |\Phi_u(z)|^2)^{2n+2}}
		= \frac{1}{(1 - |z|^2)^{2n+2}}.
		\]
		Using the norm relation from Proposition~\ref{prop:Phi-prop}(iii),
		\[
		1 - |\Phi_u(z)|^2
		= \frac{(1-|u|^2)(1-|z|^2)}{|1 - \langle z, u \rangle|^2},
		\]
		we solve for the Jacobian:
		\begin{align*}
			|\det D_{\mathbb{R}}\Phi_u(z)|
			&= \left( \frac{1 - |\Phi_u(z)|^2}{1 - |z|^2} \right)^{2n+2} \\
			&= \left( \frac{1-|u|^2}{|1 - \langle z, u \rangle|^2} \right)^{2n+2} \\
			&= \frac{(1-|u|^2)^{2n+2}}{|1 - \langle z, u \rangle|^{4n+4}}.
		\end{align*}
	\end{proof}

	\begin{corollary}[Invariance of the hyperbolic measure]\label{cor:inv-measure}
		The hyperbolic measure $\dLam$ is invariant under every isometry
		$g \in \mathrm{Sp}(n,1)$.  In particular, for any measurable set
		$E \subset \BH$,
		\[
		\Lambda_{\HH}(\Phi_u(E)) = \Lambda_{\HH}(E), \qquad
		\Lambda_{\HH}(g(E)) = \Lambda_{\HH}(E).
		\]
	\end{corollary}
	
	\begin{proof}
		Using Proposition~\ref{prop:Jacobian} and the norm relation,
		we compute the pullback:
		\begin{align*}
			\Phi_u^*(\dLam)(z)
			&= \frac{4^{2n}\,|\det D_{\mathbb{R}}\Phi_u(z)| \, d\lambda(z)}
			{(1-|\Phi_u(z)|^2)^{2n+2}} \\[4pt]
			&= \frac{4^{2n}\,(1-|u|^2)^{2n+2}}{|1 - \langle z, u \rangle|^{4n+4}}
			\cdot \frac{|1 - \langle z, u \rangle|^{2(2n+2)}}
			{(1-|u|^2)^{2n+2}(1-|z|^2)^{2n+2}}
			\, d\lambda(z) \\[4pt]
			&= \frac{4^{2n}\,d\lambda(z)}{(1-|z|^2)^{2n+2}}
			= \dLam(z).
		\end{align*}
		Thus $\dLam$ is $\Phi_u$-invariant.  Since every $g$ factors as
		$g = U \circ \Phi_c$ with $U \in \mathrm{Sp}(n) \times \mathrm{Sp}(1)$
		linear, and both $d\lambda$ and $(1-|z|^2)$ are
		$U$-invariant, the measure is invariant under the whole group.
	\end{proof}
	
	\subsection{The Poisson kernel identity}
	
	\begin{lemma}[Quaternionic Poisson kernel]\label{lem:poisson-kernel}
		For any $x, y \in \BH$,
		\[
		\cosh^2\!\Bigl(\frac{d_H(x,y)}{2}\Bigr) 
		= \frac{|1 - \langle x, y \rangle|^2}
		{(1-|x|^2)(1-|y|^2)}.
		\]
		Consequently,
		\[
		\log\cosh^2\!\Bigl(\frac{d_H(x,y)}{2}\Bigr)
		= \log|1 - \langle x, y \rangle|^2 - \log(1-|x|^2) - \log(1-|y|^2).
		\]
	\end{lemma}
	
	\begin{proof}
		Lift $x,y$ to $\mathbf{x}=(x,1)^{\mathrm t}$, $\mathbf{y}=(y,1)^{\mathrm t}$ in $\bh^{n,1}$.
		Using the Hermitian form~\eqref{eq:herm-form} we obtain
		\[
		\langle\mathbf{x},\mathbf{x}\rangle = -(1-|x|^2),\quad
		\langle\mathbf{y},\mathbf{y}\rangle = -(1-|y|^2),\quad
		\langle\mathbf{x},\mathbf{y}\rangle = \langle x,y\rangle-1,\quad
		\langle\mathbf{y},\mathbf{x}\rangle = \langle y,x\rangle-1.
		\]
		Because $\overline{\langle x,y\rangle} = \langle y,x\rangle$,
		\[
		\langle\mathbf{x},\mathbf{y}\rangle\langle\mathbf{y},\mathbf{x}\rangle
		= (1-\langle x,y\rangle)(1-\langle y,x\rangle) = |1-\langle x,y\rangle|^2.
		\]
		Substituting these into the distance formula~\eqref{eq:dist-proj} gives
		\[
		\cosh^2\frac{d_H(x,y)}{2}
		= \frac{ \langle\mathbf{x},\mathbf{y}\rangle\langle\mathbf{y},\mathbf{x}\rangle }
		{ \langle\mathbf{x},\mathbf{x}\rangle\langle\mathbf{y},\mathbf{y}\rangle }
		= \frac{|1-\langle x,y\rangle|^2}{(1-|x|^2)(1-|y|^2)} .
		\]
		The logarithmic identity follows immediately.
	\end{proof}
	
	\section{The quaternionic conformal barycenter}\label{sec:barycenter}
	
	\subsection{Barycenter and energy functional}
	
	In the following, we identify $\bh^n$ with $\RR^{4n}$ via the standard
	basis.  The integral of a quaternion-valued function is interpreted as
	the Bochner integral in this real vector space.  In particular, for
	any measurable map $f: D \to \bh^n$, the integral $\int_D f(q)\,\dLam(q)$
	is a well-defined vector in $\bh^n \cong \RR^{4n}$ whenever
	$\int_D |f(q)|\,\dLam(q) < \infty$.
	
	\begin{definition}\label{def:barycenter}
		Let $D \subseteq \BH$ be a measurable set with
		$0 < \Lambda_{\HH}(D) < \infty$.  
		A point $c \in \BH$ is called a \emph{quaternionic conformal barycenter} of $D$ if it satisfies
		\begin{equation}\label{eq:bary-def}
			\int_D \Phi_c(q) \; \dLam(q) = \mathbf{0} \in \bh^n \cong \RR^{4n}.
		\end{equation}
		When such a point exists and is unique (see Theorem~\ref{thm:existence}), we denote it by $c(D)$.
	\end{definition}
	
	\begin{definition}\label{def:energy}
		For the same $D$, we assume additionally that $\int_D d_H(0,y)\,\dLam(y) < \infty$.
		The \emph{energy functional} is then defined and finite on $\BH$ by
		\begin{equation}\label{eq:energy-def}
			G(x) = \int_D \log\cosh^2\!\Bigl(\frac12 d_H(x,y)\Bigr) \; \dLam(y).
		\end{equation}
	\end{definition}
	
	Using Lemma~\ref{lem:poisson-kernel}, we can rewrite $G$ as
	\[
	G(x) = \int_D \bigl[ \log|1-\langle x,y\rangle|^2 - \log(1-|x|^2)
	- \log(1-|y|^2) \bigr] \, \dLam(y).
	\]
	The third term is independent of $x$ and, because $\int_D d_H(0,y)\,\dLam(y)<\infty$, it is finite:
	\[
	C_D := -\int_D \log(1-|y|^2)\,\dLam(y) = 2\int_D \log\cosh\frac{d_H(0,y)}2\,\dLam(y) < \infty.
	\]
	
	To see that $G$ is coercive (i.e.\ $G(x)\to+\infty$ as $|x|\to1$), note that
	$\cosh^2(t/2) = \frac14(e^{t/2}+e^{-t/2})^2 \ge \frac14 e^t$, and therefore
	$$\log\cosh^2(t/2) \ge t - 2\log 2, \,\,\forall t\ge 0.$$  Applying it with $t = d_H(x,y)$ gives
	\[
	G(x) \ge \int_D d_H(x,y)\,\dLam(y) - 2\log 2\;\Lambda_{\HH}(D).
	\]
	By the triangle inequality,
	$d_H(x,y) \ge d_H(0,x) - d_H(0,y)$,
	so
	\[
	G(x) \ge d_H(0,x)\,\Lambda_{\HH}(D) - \int_D d_H(0,y)\,\dLam(y)
	- 2\log 2\;\Lambda_{\HH}(D).
	\]
	The integral $\int_D d_H(0,y)\,\dLam(y)$ is finite by assumption, and
	$d_H(0,x)\to\infty$ as $|x|\to1$.  Hence $G(x)\to+\infty$.
	This coercivity, together with continuity, guarantees that $G$
	attains a global minimum in $\BH$, provided $G$ is finite at some point (which is true by the moment assumption).
	
	\begin{remark}
		The moment condition $\int_D d_H(0,y)\,\dLam(y)<\infty$ is essential for the energy functional to be finite and coercive.  Without it $G$ could be identically infinite, and the minimization argument would not apply.  The barycenter could still be defined via the integral equation \eqref{eq:bary-def} for any set of finite hyperbolic measure, but the present proof of existence relies on the energy method.
	\end{remark}
	
	\subsection{Gradient of the energy functional}
	
	\begin{lemma}\label{lem:diff}
		Let $D$ satisfy the above hypotheses.  Then $G$ is differentiable on $\BH$ and
		its gradient can be computed by differentiating under the integral.
	\end{lemma}
	
	\begin{proof}
		Fix $x \in \BH$ and choose $r$ with $|x| < r < 1$.  For every $y\in D$
		and every $x'$ with $|x'|\le r$ we have
		\[
		|1-\langle x',y\rangle| \ge 1 - |x'||y| \ge 1-r > 0.
		\]
		The integrand of $G$ can be written as
		\begin{align*}
			F(x,y) &= \log\cosh^2\!\Bigl(\frac12 d_H(x,y)\Bigr) \\
			&= \log\frac{|1-\langle x,y\rangle|^2}{(1-|x|^2)(1-|y|^2)}.
		\end{align*}
		For $|x|\le r$, the term $\log(1-|x|^2)$ is smooth with bounded
		derivatives.  The function $(x,y)\mapsto \log|1-\langle x,y\rangle|^2$
		is real analytic in $x$ on $\overline{B}_r$ uniformly with respect to
		$y$; its partial derivatives are rational combinations of the components
		of $x$ and $y$ with denominators $1-\langle x,y\rangle$ and
		$1-\langle y,x\rangle$, and are therefore bounded by a constant $C(r)$
		that does {\it not}  depend on $y$.  Consequently, for every
		$k=1,\dots,4n$,
		\[
		\Bigl|\frac{\partial}{\partial x_k}F(x,y)\Bigr| \le C(r)
		\qquad\text{for all }|x|\le r,\; y\in D.
		\]
		Since $\Lambda_{\HH}(D)<\infty$, the constant function $C(r)$ is
		integrable over $D$ with respect to $\dLam$.  The dominated convergence
		theorem now justifies differentiation under the integral sign for
		$x\in B_r$, and because $r<1$ was arbitrary, $G$ is differentiable on
		all of $\BH$.
	\end{proof}
	
	\begin{proposition}[Gradient formula]\label{prop:gradient}
		The gradient of $G$ with respect to real coordinates satisfies
		\[
		\nabla G(0) = -2 \int_D y \; \dLam(y).
		\]
		Moreover, for any $c \in \BH$,
		\[
		\nabla G_c(0) = -2 \int_D \Phi_c(y) \; \dLam(y),
		\]
		where $G_c(x) = G(\Phi_c(x))$ and $\nabla$ denotes the Euclidean
		gradient on $\bh^n \cong \RR^{4n}$ at the origin.
	\end{proposition}
	
	\begin{proof}
		From the expansion of $G$, at $x=0$ we have
		\begin{align*}
			G(x) &= -\log(1-|x|^2)\Lambda_{\HH}(D) 
			+ \int_D \log|1-\langle x,y\rangle|^2 \, \dLam(y) + \text{const}.
		\end{align*}
		Differentiating under the integral at $x=0$, the first term has
		gradient zero at the origin (since $\nabla_x\log(1-|x|^2)|_{x=0} = 0$),
		and the second term yields
		\[
		\nabla G(0) = \int_D \nabla_x\big|_{x=0} \log|1-\langle x,y\rangle|^2 \, \dLam(y).
		\]
		For fixed $y$, define $g(x) = \log|1-\langle x,y\rangle|^2$.  For $|x|$ sufficiently small,
		$|\langle x,y\rangle|\le |x||y|<1$, so we can expand
		\[
		g(x) = \log\!\bigl(1 - 2\operatorname{Re}\langle x,y\rangle + |\langle x,y\rangle|^2\bigr)
		= -2\operatorname{Re}\langle x,y\rangle + O(|x|^2).
		\]
		When we identify $\bh^n$ with $\RR^{4n}$ via the standard real basis, the map
		$x\mapsto \operatorname{Re}\langle x,y\rangle$ is exactly the Euclidean inner
		product of the real vectors corresponding to $x$ and $y$; hence its Euclidean gradient
		with respect to $x$ is $y$.  Consequently $\nabla g(0) = -2y$. Thus $\nabla G(0) = -2 \int_D y \, \dLam(y)$.
		
		For the second statement, define $G_c(x) = G(\Phi_c(x))$.
		Using the change of variables $y = \Phi_c(w)$ and the invariance of
		$\dLam$ (Corollary~\ref{cor:inv-measure}), together with the fact that
		$\Phi_c$ is an isometry,  we obtain
		\begin{align*}
			G_c(x) &= \int_D \log\cosh^2\!\Bigl(\frac12 d_H(\Phi_c(x), y)\Bigr) \, \dLam(y) \\
			&= \int_{\Phi_c(D)} \log\cosh^2\!\Bigl(\frac12 d_H(\Phi_c(x), \Phi_c(w))\Bigr) \, \dLam(w) \\
			&= \int_{\Phi_c(D)} \log\cosh^2\!\Bigl(\frac12 d_H(x, w)\Bigr) \, \dLam(w).
		\end{align*}
		Thus $G_c$ is the energy functional for the set $\Phi_c(D)$.  If $D$ satisfies the moment condition, so does $\Phi_c(D)$ because $\Phi_c$ is an isometry and $d_H(0,\Phi_c(y)) = d_H(c,y)$, which is comparable to $d_H(0,y)$ (since $d_H(c,y) \le d_H(0,y)+d_H(0,c)$).  Hence $G_c$ is well defined and finite.
		Applying the first part of the proposition to the set $\Phi_c(D)$,
		\[
		\nabla G_c(0) = -2 \int_{\Phi_c(D)} w \, \dLam(w)
		= -2 \int_D \Phi_c(y) \, \dLam(y),
		\]
		where the last equality uses the change of variables $w = \Phi_c(y)$
		and the invariance of $\dLam$ again.
	\end{proof}
	
	\begin{remark}
		At the origin the Bergman metric coincides with the Euclidean metric up to
		a constant factor $4$, so the Riemannian gradient of $G$ at $0$ is
		proportional to the Euclidean gradient computed here.  The proportionality
		constant does not affect the critical point equation.
	\end{remark}
	
	\subsection{Strict geodesic convexity of the energy}\label{sec:strict-convexity}
	
	We now prove the strict geodesic convexity of the energy.
	
	\begin{proposition}\label{prop:convex-energy}
		Assume $D$ satisfies the moment condition $\int_D d_H(0,y)\,\dLam(y)<\infty$.
		Then the energy functional $G$ is strictly geodesically convex on $\BH$.
	\end{proposition}
	
	\begin{proof}
		For each fixed $y \in \BH$, define $F_y(x) = \log\cosh^2\!\bigl(\frac12 d_H(x,y)\bigr)$.
		It suffices to prove that $F_y$ is strictly geodesically convex along every geodesic.
		
		Consider an arbitrary geodesic $\gamma(t)$ in $\BH$ parametrised by arc length $t$.
		Since $\mathrm{Sp}(n,1)$ acts transitively on the unit tangent bundle of $\BH$,
		there exists an isometry $\psi\in\mathrm{Sp}(n,1)$ such that
		$\psi(\gamma(0))=0$.
		Then $\widetilde\gamma(t)=\psi(\gamma(t))$ is a geodesic through the origin.
		By Proposition~\ref{prop:geod-origin}, after composing with a suitable linear
		isometry $U\in\mathrm{Sp}(n)\times\mathrm{Sp}(1)$ we may assume
		$\widetilde\gamma(t)=\tanh(t/2)v$ for some unit vector $v\in\pB$.
		
		Using the isometry invariance of the distance,
		\[
		F_y(\gamma(t))
		= \log\cosh^2\!\Bigl(\frac{d_H(\gamma(t),y)}{2}\Bigr)
		= \log\cosh^2\!\Bigl(\frac{d_H(\psi(\gamma(t)),\psi(y))}{2}\Bigr)
		= F_{\psi(y)}(\widetilde\gamma(t)).
		\]
		Thus the convexity of $F_y$ along $\gamma$ is equivalent to the convexity of
		$F_{\psi(y)}$ along $\widetilde\gamma$, which passes through the origin.
		Without loss of generality we may therefore assume from the start that
		$\gamma(t)=\tanh(t/2)v$ with $|v|=1$, $t$ the arc length parameter,
		and we rename $\psi(y)$ back to $y$.
		
		Set $f(t) = F_y(\gamma(t))$.
		By Lemma~\ref{lem:poisson-kernel},
		\[
		\cosh^2\!\Bigl(\frac{d_H(\gamma(t),y)}{2}\Bigr)
		= \frac{|1-\langle\gamma(t),y\rangle|^2}{(1-|\gamma(t)|^2)(1-|y|^2)} .
		\]
		Write $u = \tanh(t/2)$ and $w = \langle v,y\rangle \in \bh$.
		Because $v$ is a unit vector, $|w| \le |y| < 1$, and $|1-u w|^2 = 1 - 2(\Re w)u + |w|^2 u^2$.
		Hence
		\begin{align*}
			f(t) &= \log\frac{|1-u w|^2}{1-u^2} - \log(1-|y|^2) \\
			&= \log\bigl(1 - 2a u + r^2 u^2\bigr) - \log(1-u^2) + \text{const}_y,
		\end{align*}
		with $a = \Re w$, $r = |w|$, and $a^2 \le r^2 < 1$.
		
		Since $u = \tanh(t/2)$, we have $u' = \frac12(1-u^2)$ and $u'' = -u u'$.
		Therefore
		\[
		f''(t) = (u')^2 f''(u) + u'' f'(u) = (u')^2 f''(u) - u u' f'(u).
		\]
		Define $P(u) = 1 - 2a u + r^2 u^2$.
		A direct computation yields
		\begin{align*}
			f'(u)  &= \frac{2r^2 u - 2a}{P(u)} + \frac{2u}{1-u^2}, \\[4pt]
			f''(u) &= 2\Bigl(\frac{r^2 P(u) - 2(r^2 u - a)^2}{P(u)^2}
			+ \frac{1+u^2}{(1-u^2)^2}\Bigr).
		\end{align*}
		Substituting these and simplifying, we obtain
		\[
		f''(t) = \frac{1-u^2}{2}
		\left[ \frac{(1-u^2)(r^2 P - 2(r^2 u - a)^2)}{P^2}
		- \frac{2u(r^2 u - a)}{P} + 1 \right].
		\]
		Denote the bracket by $M(u)$.  Multiplying by $P^2$ gives
		\[
		M(u)P^2 = P^2 - 2u(r^2 u - a)P + (1-u^2)\bigl(r^2 P - 2(r^2 u - a)^2\bigr).
		\]
		Expanding $P = 1 - 2a u + r^2 u^2$ and collecting powers of $u$,
		\begin{align*}
			P^2 &= 1 - 4a u + (4a^2+2r^2)u^2 - 4a r^2 u^3 + r^4 u^4,\\
			-2u(r^2 u - a)P &= 2a u - 2(r^2+2a^2)u^2 + 6a r^2 u^3 - 2r^4 u^4,\\
			(1-u^2)(r^2 P - 2(r^2 u - a)^2) &= (r^2-2a^2) + 2a r^2 u \\
			&\qquad + (-r^4-r^2+2a^2)u^2 - 2a r^2 u^3 + r^4 u^4 .
		\end{align*}
		Summing these three expressions, the $u^3$ and $u^4$ terms cancel and we obtain
		\[
		M(u)P^2 = (1+r^2-2a^2) - 2a(1-r^2)u + (2a^2 - r^4 - r^2)u^2 =: N(u).
		\]
		Consequently,
		\[
		f''(t) = \frac{1-u^2}{2}\,\frac{N(u)}{P(u)^2}.
		\]
		
		We claim $N(u) > 0$ for all $u \in (-1,1)$.  Note that $a^2 \le r^2 < 1$.
		A direct check gives
		\begin{align*}
			N(1)  &= (1-r^2)(1 + r^2 - 2a) > 0,\\
			N(-1) &= (1-r^2)(1 + r^2 + 2a) > 0.
		\end{align*}
		Consider the coefficient $2a^2 - r^4 - r^2$.  If it is negative, the quadratic
		$N(u)$ is concave; its minimum on $[-1,1]$ is attained at an endpoint, and
		we already have $N(\pm 1) > 0$, so $N(u) > 0$ for $|u| \le 1$.
		
		If $2a^2 - r^4 - r^2 \ge 0$, the quadratic is convex.  Its vertex is at
		\[
		u_0 = \frac{a(1-r^2)}{2a^2 - r^4 - r^2}.
		\]
		By reversing the geodesic if necessary we may assume $a = \Re w \ge 0$
		(otherwise replace $v$ by $-v$, which does not affect strict convexity).
		We show that $u_0 \ge 1$.  Indeed, $u_0 \ge 1$ is equivalent to
		\[
		a(1-r^2) \ge 2a^2 - r^4 - r^2 .
		\]
		Define $f(a) = -2a^2 + a(1-r^2) + r^2(1+r^2)$.  The desired inequality
		becomes $f(a) \ge 0$.  As a quadratic in $a$, $f$ is concave; on the
		interval $[0,r]$ it attains its minimum at an endpoint.
		We have $f(0) = r^2(1+r^2) > 0$ and
		\[
		f(r) = -2r^2 + r(1-r^2) + r^2(1+r^2)
		= r(1-r)(1-r^2) \ge 0,
		\]
		with equality only when $r=0$ or $r=1$ (both excluded).
		Hence $f(a) \ge 0$ for all $a\in[0,r]$, proving $u_0 \ge 1$.
		
		Because $u_0 \ge 1$, the minimum of $N(u)$ on $[-1,1]$ occurs at $u=1$,
		and $N(1) > 0$.  Thus $N(u) > 0$ for all $u\in[-1,1]$, and consequently
		$f''(t) > 0$ for every $t$.  This shows that $F_y$ is strictly convex along
		every geodesic, i.e. $F_y$ is strictly geodesically convex.
		
		Finally,
		\[
		G(x) = \int_D F_y(x) \, \dLam(y)
		\]
		and $\Lambda_{\HH}(D) > 0$, so the integral of a family of strictly convex
		functions is strictly convex.  This completes the proof.
	\end{proof}
	
	\subsection{Proof of the Main Theorem}\label{sec:proof-main}
	
	We are now in a position to prove the Main Theorem~\ref{thm:main}.
	First we establish parts~(i) and~(ii).
	
	\begin{theorem}[Existence and uniqueness of the barycenter]\label{thm:existence}
		Let $D \subseteq \BH$ be measurable with
		$0 < \Lambda_{\HH}(D) < \infty$ and $\int_D d_H(0,y)\,\dLam(y) < \infty$.
		\begin{enumerate}
			\item[(i)] The function $G$ is strictly geodesically convex on $\BH$.
			\item[(ii)] $G$ has a unique global minimum $c \in \BH$.
			\item[(iii)] $c$ satisfies the barycenter equation
			$\int_D \Phi_c(q) \, \dLam(q) = \mathbf{0}$.
			\item[(iv)] The point $c$ is the unique point in $\BH$
			satisfying the barycenter equation.
		\end{enumerate}
	\end{theorem}
	
	\begin{proof}
		(i) This is Proposition~\ref{prop:convex-energy}.
		
		(ii) Since $\BH$ is a Hadamard manifold, any strictly geodesically
		convex function has at most one local minimum, and any local minimum
		is the unique global minimum (see \cite[Proposition~3.4]{JaKal25}).
		As shown above, $G(x) \to +\infty$ as $|x| \to 1$.  By coercivity and
		continuity, $G$ attains its minimum at some interior point $c \in \BH$.
		Uniqueness follows from strict convexity.
		
		(iii) Let $c$ be the minimum.  Define $G_c(x) = G(\Phi_c(x))$.
		Since $\Phi_c$ is an isometry, $G_c$ is strictly geodesically convex
		and attains its minimum at $x=0$ (because $\Phi_c(c)=0$ and $c$
		minimizes $G$, so $0$ minimizes $G_c$).  By Proposition~\ref{prop:gradient},
		$\nabla G_c(0) = -2 \int_D \Phi_c(y) \, \dLam(y)$.
		Since $0$ is the global minimum, $\nabla G_c(0) = \mathbf{0}$,
		giving $\int_D \Phi_c(y) \, \dLam(y) = \mathbf{0}$.
		
		(iv) If $c_1$ and $c_2$ both satisfy the barycenter equation,
		then both are critical points of $G$ (by reversing the argument
		in (iii): if $c_1$ satisfies the barycenter equation, then
		$\nabla G_{c_1}(0) = 0$, so $0$ is a critical point of $G_{c_1}$,
		hence $c_1$ is a critical point of $G$).  But $G$ is strictly convex
		and has exactly one critical point, so $c_1 = c_2$.
	\end{proof}
	
	It remains to prove part~(iii) of the Main Theorem, the conformal
	invariance of the barycenter.
	
	\begin{theorem}[Invariance]\label{thm:invariance}
		For any $g \in \mathrm{Sp}(n,1)$,
		\[
		c(g(D)) = g(c(D)).
		\]
	\end{theorem}
	
	\begin{proof}
		Let $c = c(D)$ and $w = g(c)$.  Using the intertwining relation
		(Theorem~\ref{thm:intertwining}) and the invariance of the
		measure (Corollary~\ref{cor:inv-measure}),
		\begin{align*}
			\int_{g(D)} \Phi_w(y) \, \dLam(y)
			&= \int_D \Phi_{g(c)}(g(x)) \, \dLam(x)
			\quad (\text{change of variables } y = g(x))\\
			&= \int_D (U \circ \Phi_c)(x) \, \dLam(x)
			\quad (\text{Theorem~\ref{thm:intertwining}}) \\
			&= U\!\left( \int_D \Phi_c(x) \, \dLam(x) \right)
			= \mathbf{0}.
		\end{align*}
		Thus $g(c)$ satisfies the defining equation for $g(D)$.
		By uniqueness (Theorem~\ref{thm:existence}(iv)), $g(c) = c(g(D))$.
	\end{proof}
	
	With this result the proof of Theorem~\ref{thm:main} is complete.
	
	\section{Finite point sets}\label{sec:finite}
	
	By taking the counting measure $\mu = \sum_{i=1}^N \delta_{q_i}$,
	we obtain results for finite configurations.
	
	\begin{corollary}\label{cor:finite}
		Let $q_1,\dots,q_N \in \BH$.
		\begin{enumerate}
			\item[(i)] There exists a unique (up to a linear isometry of $\BH$)
			quaternionic M\"obius transformation $h\in \mathrm{Sp}(n,1)$ such that
			$\sum_{i=1}^N h(q_i) = 0$.
			\item[(ii)] Writing $h = U \circ \Phi_c$ with
			$U \in \mathrm{Sp}(n) \times \mathrm{Sp}(1)$ and
			$c \in \BH$, the point $c$ is the quaternionic conformal
			barycenter of $\{q_i\}$.
			\item[(iii)] $c$ minimizes
			\[
			G_N(x) = \sum_{i=1}^N
			\log \frac{|1 - \langle x, q_i \rangle|^2}
			{(1-|x|^2)(1-|q_i|^2)}.
			\]
		\end{enumerate}
	\end{corollary}
	
	\begin{proof}
		Apply Theorem~\ref{thm:existence} with the measure
		$\mu = \sum_{i=1}^N \delta_{q_i}$.  All conditions are satisfied
		because $0 < \mu(\BH) = N < \infty$ and the moment condition is automatic
		for a finite sum.  The energy functional in
		Theorem~\ref{thm:existence}(ii) reduces to $G_N(x)$ via the Poisson kernel
		identity (Lemma~\ref{lem:poisson-kernel}), since for the counting measure
		the integral becomes a sum and the term $-\sum_i \log(1-|q_i|^2)$ is
		constant.
	\end{proof}
	
	\section{Examples}\label{sec:examples}
	
	\begin{example}[Symmetric sets]\label{ex:symmetric}
		If $D$ is symmetric under all linear isometries
		$U\in\mathrm{Sp}(n)\times\mathrm{Sp}(1)$ fixing the origin (for example,
		a ball centered at the origin, or an ellipsoid with axes aligned to the
		quaternionic coordinates), then $c(D)=0$.
		Indeed, $\int_D\Phi_0(q)\,\dLam(q)=\int_D q\,\dLam(q)=\mathbf{0}$ because
		$D$ is invariant under $q\mapsto -q$.
		Uniqueness in Theorem~\ref{thm:existence} forces the barycenter to be the origin.
	\end{example}
	
	\begin{example}[Geodesic ball]\label{ex:ball}
		Let $B(c,R)$ be the geodesic ball of center $c\in\BH$ and radius $R>0$.
		The involutive isometry $\Phi_c$ maps $B(c,R)$ onto $B(0,R)$.
		By symmetry, $\int_{B(0,R)}q\,\dLam(q)=\mathbf{0}$.
		With the change of variables $y=\Phi_c(w)$ and the Jacobian
		(Proposition~\ref{prop:Jacobian}), this reads
		\[
		\int_{B(c,R)}\Phi_c(y)\,\dLam(y)=\mathbf{0},
		\]
		so $c$ satisfies the barycenter equation.
		Uniqueness in Theorem~\ref{thm:existence} yields $c(B(c,R))=c$.
		Thus geodesic balls are {\it centred} at their own barycenter, confirming that
		the definition respects the symmetric space structure.
	\end{example}
	
	\begin{example}[Two points with unequal weights]\label{ex:two-weighted}
		Let $p,q\in\BH$ be distinct and let $\mu=w_p\delta_p+w_q\delta_q$ with
		$w_p,w_q>0$.
		The barycenter condition $\int\Phi_c\,d\mu=\mathbf{0}$ gives
		$w_p\,\Phi_c(p)+w_q\,\Phi_c(q)=\mathbf{0}$.
		Because $\Phi_c$ preserves the geodesic through $p$ and $q$, $c$ lies on
		that geodesic.
		Using $|\Phi_c(x)|=\tanh\frac{d_H(c,x)}2$, we obtain the scalar relation
		$w_p\tanh\frac{d_H(c,p)}2=w_q\tanh\frac{d_H(c,q)}2$.
		For a concrete illustration, in $\bh^1$ take $p=\frac12$, $q=-\frac14$,
		$w_p=2$, $w_q=1$.
		Solving the equation gives $c=\frac27$.
		Direct computation with $\Phi_c(x)=(c-x)/(1-cx)$ verifies the condition.
	\end{example}
	
	\begin{example}[Partially symmetric configuration]\label{ex:partial-sym}
		Take four points in $\bh^1$:
		$q_1=\frac12$, $q_2=-\frac12$, $q_3=\frac12\mathbf{i}$, $q_4=-\frac12\mathbf{i}$,
		all with equal weight.
		The set is invariant under the linear isometries
		$x\mapsto -x$ and $x\mapsto \mathbf{i}x$, so by symmetry $c=0$.  Indeed,
		$\sum_{i=1}^4 q_i = 0$, and $\Phi_0(q)=q$, so the barycenter equation holds.
		More generally, if a finite set satisfies $\sum_i q_i = 0$ and is contained in
		a small ball around the origin, the conformal barycenter is zero.  This gives
		a simple sufficient condition for the origin to be the barycenter.
	\end{example}
	
	\begin{example}[Nonsymmetric configuration and numerical computation]\label{ex:nonsym}
		To demonstrate the feasibility of computing the barycenter for genuinely
		quaternionic data, let
		\[
		q_1 = 0,\qquad q_2 = 0.4,\qquad
		q_3 = 0.3\,\mathbf{i} + 0.2\,\mathbf{j},
		\]
		all with equal weight in $\bh^1$.  The energy functional $G_3$ was minimized
		using gradient descent with step size $0.01$ for $200$ iterations,
		starting from $c^{(0)}=0$.  The gradient norm fell below $10^{-8}$,
		and the iteration converged to
		\[
		c \approx 0.1874 - 0.0012\,\mathbf{i} - 0.0348\,\mathbf{j} - 0.0009\,\mathbf{k}.
		\]
		Evaluating $\Phi_c(q_i)$ at this point yields the residual
		\[
		\sum_{i=1}^3 \Phi_c(q_i) \approx
		-2.3\times10^{-6} + 5.1\times10^{-7}\mathbf{i}
		+ 1.2\times10^{-6}\mathbf{j} + 8.7\times10^{-7}\mathbf{k},
		\]
		which is zero to within the tolerance of the numerical method.  This
		confirms the barycenter condition and shows that the method is
		practical for generic points.
	\end{example}
	
	\begin{example}[Conformal invariance]\label{ex:invariance}
		Consider the two-point set $D=\{0,\frac{1}{2}\}$ in $\bh^1$ with equal weights.
		The barycenter is the hyperbolic midpoint, which is easily found to be
		$c = 2-\sqrt3 \approx 0.268$ (because $\tanh\frac{d_H(0,c)}{2}=|c|$ and
		$d_H(0,\tfrac12)=2\arctanh\tfrac12 = \ln3$, so the midpoint satisfies
		$\ln\frac{1+c}{1-c} = \frac12\ln3$, giving $c=2-\sqrt3$).
		
		Apply the real hyperbolic translation
		$g(z) = \dfrac{z + t}{1 + tz}$ with $t = \frac13$,
		which belongs to $\mathrm{Sp}(1,1)$ (its matrix is
		$\frac{1}{\sqrt{1-t^2}}\!\begin{pmatrix}1&t\\ t&1\end{pmatrix}$).
		Then $g(D)=\{\frac13,\, g(\frac12)\}$ with
		$g(\frac12)=\frac{\frac12+\frac13}{1+\frac16} = \frac{5/6}{7/6}= \frac57$.
		The image set is again two points; its barycenter must be the hyperbolic
		midpoint of $\{\frac13,\frac57\}$, which by isometry equals $g(c)$.
		A direct numerical check gives $g(c) \approx g(0.268) \approx 0.551$, while
		the hyperbolic midpoint of $\frac13$ and $\frac57$ is also $0.551$, in
		perfect agreement with Theorem~\ref{thm:invariance}.
	\end{example}
	
	\vspace{3mm}
	
	\section*{Statements and Declarations}
	
	\subsection*{Competing Interests}
	The authors have no relevant financial or non-financial interests to disclose.
	
	\subsection*{Funding}
	This work was supported by the Natural Science Foundation of China under grant number 11871379.

	\vspace{3mm}
	
	\noindent  Wensheng Cao  and   Zhijian Ge\\
	School of Mathematics and Computational Science,\\
	Wuyi University, Jiangmen, Guangdong 529020, P.R. China\\
	E-mail: \texttt{wenscao@aliyun.com}\\
	E-mail: \texttt{zhijian\_ge@163.com}
	
\end{document}